\documentclass[twoside,leqno,twocolumn]{article}

\usepackage{graphicx}
\usepackage{cancel}
\usepackage{rotating}
\usepackage{array}
\usepackage{theoremref}
\usepackage{amsmath}
\usepackage{amssymb}
\usepackage{url}
\usepackage{cite}
\usepackage{breakurl}
\usepackage{multirow}
\usepackage{indentfirst}
\usepackage{mdwlist}
\usepackage{setspace}
\usepackage{float}

\usepackage[titletoc]{appendix}

\usepackage[letterpaper]{geometry}
\usepackage{ltexpprt}

\usepackage[T1]{fontenc}
\usepackage{mathtools}

\allowdisplaybreaks

\begin{document}

\title{\Large Optimal construction of a layer-ordered heap\thanks{Supported by grant
    number 1845465 from the National Science Foundation.}}

\author{Jake Pennington\thanks{University of Montana Department of Mathematics}
  \and
  Patrick Kreitzberg$^\dagger$
  \and
  Kyle Lucke\thanks{University of Montana Department of Computer Science}
  \and
  Oliver Serang$^\ddagger$\thanks{Corresponding Author, Email: oliver.serang@umontana.edu}
}

\maketitle



\fancyfoot[R]{\scriptsize{Copyright \textcopyright\ 2020\\
Copyright for this paper is retained by authors}}


\begin{abstract}
  \noindent The layer-ordered heap (LOH) is a simple, recently
  proposed data structure used in optimal selection on $X+Y$, the
  algorithm with the best known runtime for selection on
  $X_1+X_2+\cdots+X_m$, and the fastest method in practice for
  computing the most abundant isotope peaks in a chemical
  compound. Here, we introduce a few algorithms for constructing LOHs,
  analyze their complexity, and demonstrate that one algorithm is
  optimal for building a LOH of any rank $\alpha$. These results are
  shown to correspond with empirical experiments of runtimes when
  applying the LOH construction algorithms to a common task in machine
  learning.
\end{abstract}

\clearpage

\section{Introduction}
Layer-ordered heaps (LOHs) are used in algorithms that perform optimal
selection on $A+B$\cite{serang:optimal}, algorithms with the best
known runtime for selection on
$X_1+X_2+\cdots+X_m$\cite{kreitzberg:selection}, and the fastest known
method for computing the most abundant isotopes of a
compound\cite{kreitzberg:fast}.

A LOH of rank $\alpha$ consists of $\ell$ layers of values $L_0 \leq
L_1 \leq L_2 \leq \cdots \leq L_{\ell-1}$, where each $L_i$ is an
unordered array and the ratio of the sizes,
$\frac{|L_{i+1}|}{|L_{i}|}$, tends to $\alpha$ as the index, $i$,
tends to infinity. One possible way to achieve this is to have the
exact size of each layer, $|L_i|$, be $p_{i}-p_{i-1}$ where $p_{i}$,
the $i^{th}$ pivot, is calculated as $p_{i} =
\left\lceil\sum_{j=0}^{i}\alpha^j\right\rceil $. The size of the last
layer is the difference between the size of the array and the last
pivot. Figure~\ref{fig:loh} depicts a LOH of rank
$\alpha=2$. 

Soft heaps\cite{chazelle:soft} are qualitatively similar in that they
achieve partial ordering in theoretically efficient time; however, the
disorder from a soft heap occurs in a less regular manner, requiring
client algorithms to cope with a bounded number of ``corrupt''
items. Furthermore, they are less efficient in practice because of the
discontiguous data structures and the greater complexity of
implementation.

Throughout this paper, the process of constructing a LOH of rank
$\alpha$ from an array of length $n$ will be denoted ``LOHification.''
While LOHify with $\alpha=1$ is equivalent to comparison sort and
$\alpha\gg 1$ can be performed in $O(n)$\cite{kreitzberg:selection},
the optimal runtime for an arbitrary $\alpha$ is unknown. Likewise, no
optimal LOHify algorithm is known for arbitrary $\alpha$.

Here we derive a lower bound runtime for LOHification, describe a
few algorithms for LOHification, prove their runtimes, and demonstrate
optimality of one method for any $\alpha$. We then demonstrate the
practical performance of these methods on a non-parametric stats test.

\begin{figure}[H]
  \centering
  \includegraphics[width=.45\textwidth]{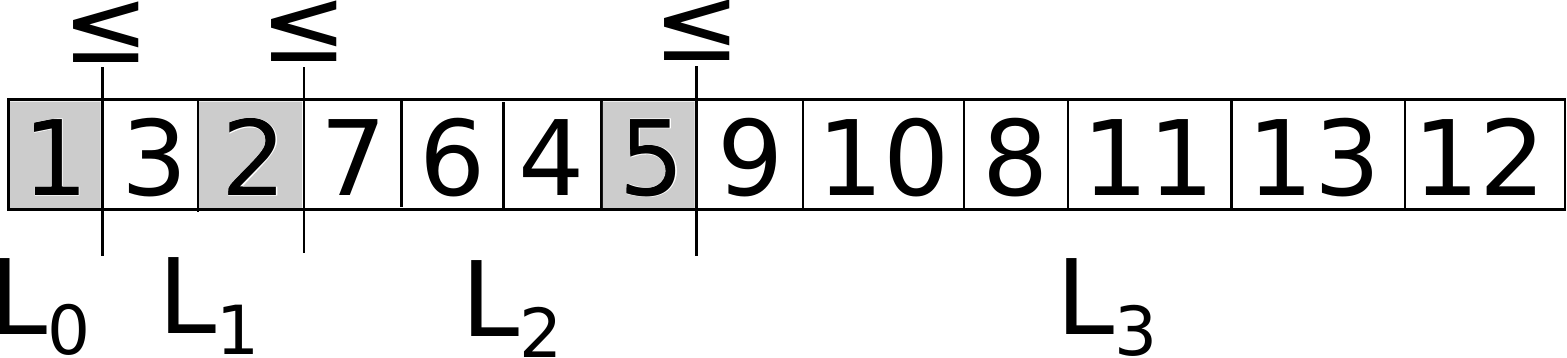}
  \caption{\textbf{A LOH of rank 2} Pivot indices are shaded in
    gray. Notice that the last layer is not full.
  \label{fig:loh}}
\end{figure}

\section{Methods}
\subsection{A lower bound on LOHification}
In this section we will prove an asymptotic lower bound on the complexity
of constructing a LOH in terms of $n$ and $\alpha$ by first proving bounds
on variables and then using those to bound the process as a whole. 
\subsubsection{Bounds on variables}
\begin{lemma}[Upper bound on number of layers]\thlabel{lemma:upper-bound-on-number-layers}
  An upper bound on the number of layers, $\ell$, in a LOH of $n$ elements is
  $\log_{\alpha}(n\cdot(\alpha-1)+1)+1$.
\end{lemma}
\begin{proof}
Because the final pivot can be no more than $n$, the size of our array,
we have the following inequality:
\begin{eqnarray*}
  \left\lceil\sum_{i=0}^{\ell-2}\alpha^{i}\right\rceil &\leq& n\\
  \sum_{i=0}^{\ell-2}\alpha^{i} &\leq& n\\
  \frac{\alpha^{\ell-1}-1}{\alpha-1} &\leq& n\\
  \alpha^{\ell-1}-1 &\leq& n\cdot(\alpha-1)\\
  \alpha^{\ell-1} &\leq& n\cdot(\alpha-1)+1\\
  \ell-1 &\leq& \log_{\alpha}(n\cdot(\alpha-1)+1)\\
  \ell &\leq& \log_{\alpha}(n\cdot(\alpha-1)+1)+1
\end{eqnarray*}
\end{proof}
\begin{lemma}[Lower bound on number of layers]\thlabel{lemma:lower-bound-on-number-layers}
  A lower bound on the number of layers, $\ell$, in a LOH of $n$ elements is
  $\log_{\alpha}(n\cdot(\alpha-1)+1)$.
\end{lemma}
\begin{proof}
  Because an additional pivot (after the final pivot) must be more than $n$,
  the size of our array, we have the following inequality:
\begin{eqnarray*}
  \left\lceil\sum_{i=0}^{\ell-1}\alpha^{i}\right\rceil &>& n\\
  \sum_{i=0}^{\ell-1}\alpha^{i} &\geq& n \text{, because $n$ is discrete;} \\
  \frac{\alpha^{\ell}-1}{\alpha-1} &\geq& n\\
  \alpha^{\ell}-1 &\geq& n\cdot(\alpha-1)\\
  \alpha^{\ell} &\geq& n\cdot(\alpha-1)+1\\
  \ell &\geq& \log_{\alpha}(n\cdot(\alpha-1)+1)\\
  \ell &>& \log_{\alpha}(n\cdot(\alpha-1))
\end{eqnarray*}
\end{proof}

\begin{lemma}[Asypmtotic number of layers]\thlabel{lemma:asymptotic-number-of-layers}
  The number of layers as $n$ grows is asymptotic to $\log_{\alpha}(n\cdot(\alpha-1)+1)$.
\end{lemma}
\begin{proof}
For $\alpha=1$, the number of layers is n.
\begin{eqnarray*}
  && \lim\limits_{\alpha\rightarrow 1} \log_{\alpha}(n\cdot(\alpha-1)+1) \\
  &=& \lim\limits_{\alpha\rightarrow 1} \frac{\log(n\cdot(\alpha-1)+1)}{\log(\alpha)} \\
  &=& \lim\limits_{\alpha\rightarrow 1} \frac{\left(\frac{n}{n\cdot(\alpha-1)+1}\right)}{\left(\frac{1}{\alpha}\right)} \text{~by L'H{\^o}pital's rule} \\
  &=& \lim\limits_{\alpha\rightarrow 1} \frac{n\cdot\alpha}{n\cdot(\alpha-1)+1} \\
  &=& n
\end{eqnarray*}
\indent For $\alpha>1$ we know $\log_{\alpha}(n\cdot(\alpha-1)+1) \leq \ell \leq \log_{\alpha}(n\cdot(\alpha-1)+1)+1$.
\begin{eqnarray*}
  && \lim\limits_{n\rightarrow\infty} \frac{\log_{\alpha}(n\cdot(\alpha-1)+1)+1}{\log_{\alpha}(n\cdot(\alpha-1)+1)} \\
  &=& \lim\limits_{n\rightarrow\infty} \frac{\log_{\alpha}(n\cdot(\alpha-1)+1)}{\log_{\alpha}(n\cdot(\alpha-1)+1)} + \frac{1}{\log_{\alpha}(n\cdot(\alpha-1)+1)} \\
  &=& 1 + 0 \\
  &=& 1
\end{eqnarray*}
\end{proof}
\begin{lemma}[Upper bound on size of layers]\thlabel{lemma:upper-bound-on-the-size-of-a-layer} 
  An upper bound on the size of layer $i$ is
  $|L_i|\leq \lceil\alpha^i\rceil$.
\end{lemma}
\begin{proof}
$|L_i|$, as defined above, can be calculated by:
\begin{eqnarray*}
  |L_i| &=& p_i - p_{i+1} \\
  &=& \left\lceil\sum_{j=0}^{i}\alpha^j\right\rceil - \left\lceil\sum_{j=0}^{i-1}\alpha^j\right\rceil \\
  &\leq& \lceil\alpha^i\rceil+\left\lceil\sum_{j=0}^{i-1}\alpha^j\right\rceil - \left\lceil\sum_{j=0}^{i-1}\alpha^j\right\rceil \\
  &\leq& \lceil\alpha^i\rceil
\end{eqnarray*}
\end{proof}
\subsubsection{Lower bound of LOHification}
Here we will show that, for any $\alpha>1$, LOHification is in
$\Omega\left(n\log(\frac{1}{\alpha-1})+\frac{n\cdot\alpha\cdot\log(\alpha)}
{\alpha-1}\right)$.

\begin{theorem}\thlabel{thm:lower-bound-on-lohify}
  $\forall\alpha>1$, LOHification $\in\Omega\left(n\log(\frac{1}{\alpha-1})+\frac{n\cdot\alpha\cdot\log(\alpha)}{\alpha-1}\right)$
\end{theorem}
\begin{proof}
If $\alpha=1$, we are sorting, which is known to be in $\Omega(n\log(n))$.
Hence, for the following derivation, we shall assume that $\alpha>1$.
From $n!$ possible unsorted arrays, LOHification produces one of
$|L_0|!\cdot |L_1|! \cdots |L_{\ell-1}|!$ possible valid results; hence,
using an optimal decision tree, $r(n)$ is in $\Omega\left(
\log_2\left(\binom{n}{|L_0|,|L_1|,\ldots,|L_{\ell-1}|}\right) \right)$; hence,
\begin{eqnarray*}
  && r(n)\\
  &\in& \Omega\left(\log\left(\frac{n!}{\prod_{i=0}^{\ell-1}(|L_i|!)}\right)\right)\\
  &=& \Omega\left(n\log(n)-\sum_{i=0}^{\ell-1}\log(|L_i|!)\right)\\
  &=& \Omega\left(n\log(n)-\sum_{i=0}^{\ell-1}\log(\lceil\alpha^{i}\rceil!)\right) \text{by \thref{lemma:upper-bound-on-the-size-of-a-layer}} \\
  &=& \Omega\left(n\log(n)-\sum_{i=0}^{\ell-1}\lceil\alpha^{i}\rceil\cdot\log(\lceil\alpha^{i}\rceil)\right) \\
  && \text{(since $\log(n!) \in \Theta(n\log(n))$)}\\
  &=& \Omega\left(n\log(n)-\sum_{i=0}^{\ell-1}\alpha^{i}\cdot\log(\alpha^{i})\right) \text{by \thref{lemma:asymptotic-of-log-of-ceiling}}\\
  &=& \Omega\left(n\log(n)-\sum_{i=0}^{\ell-1}i\cdot\alpha^i\log(\alpha)\right)\\
  &=& \Omega\left(n\log(n)-\log(\alpha)\cdot\sum_{i=0}^{\ell-1}i\cdot\alpha^i\right)\\
  &=& \Omega\left(n\log(n)-\log(\alpha)\cdot \right. \\
  && \left.\left(\frac{\alpha^{\ell+1}\cdot(\ell-1)+\alpha-\alpha^{\ell}\cdot\ell}{(\alpha-1)^2}\right)\right)\\
  &=& \Omega\left(n\log(n)-\log(\alpha) \right. \\
  && \cdot \left(\frac{(\alpha^{\log_{\alpha}(n\cdot(\alpha-1)+1)+1}\cdot\log_{\alpha}(n\cdot(\alpha-1)+1)-1)}{(\alpha-1)^2} \right.\\
  && +\frac{\alpha}{(\alpha-1)^2}  \\
  && - \left. \left. \frac{\alpha^{\log_{\alpha}(n\cdot(\alpha-1)+1)}\cdot\log_{\alpha}(n\cdot(\alpha-1)+1)}{(\alpha-1)^2} \right)\right) \\
  && \text{by \thref{lemma:asymptotic-number-of-layers}}\\
  &=& \Omega\left(n\log(n)-\log(\alpha) \right. \cdot \\
  && \left(\frac{(n\cdot(\alpha-1)+1)\cdot\alpha\cdot(\log_{\alpha}(n\cdot(\alpha-1)+1)-1)}{(\alpha-1)^2} \right. \\
  && +\frac{\alpha}{(\alpha-1)^{2}}  \\
  && - \left. \left. \frac{(n\cdot(\alpha-1)+1)\cdot\log_{\alpha}(n\cdot(\alpha-1)+1)}{(\alpha-1)^{2}}\right)\right) \\
  &=& \Omega\left(n\log(n)-\right( \\
  && \frac{(n\cdot(\alpha-1)+1)\cdot\alpha\cdot(\log(n\cdot(\alpha-1)+1)-\log(\alpha))}{(\alpha-1)^2}\\
  && + \frac{\alpha\log(\alpha)}{(\alpha-1)^2} \\
  && - \left.\left.\frac{(n\cdot(\alpha-1)+1)\log (n\cdot(\alpha-1)+1)}{(\alpha-1)^2}\right)\right) \\
  &=& \Omega(n\log(n)- \\
  && \left(\frac{(n\cdot(\alpha-1)+1)\cdot\alpha\log (n\cdot(\alpha-1)+1)}{(\alpha-1)^2}\right.\\
  && - \left.\left. \frac{(n\cdot(\alpha-1)+1)\cdot\alpha\log(\alpha)}{(\alpha-1)^2} + \frac{\alpha\log(\alpha)}{(\alpha-1)^2}\right.\right.\\
  && - \left.\left.\frac{(n\cdot(\alpha-1)+1)\cdot\log (n\cdot(\alpha-1)+1)}{(\alpha-1)^2}\right)\right)\\
  &=& \Omega(n\log(n)- \\
  && \left(\frac{(n\cdot(\alpha-1)+1)\cdot\log (n\cdot(\alpha-1)+1)\cdot(\alpha-1)}{(\alpha-1)^2}\right. \\
  && + \left.\left. \frac{\alpha\log(\alpha)}{(\alpha-1)^2}- \frac{(n\cdot(\alpha-1)+1)\cdot\alpha\log(\alpha)}{(\alpha-1)^2}\right)\right) \\
  &=& \Omega(n\log(n)- \\
  && \frac{(n\cdot(\alpha-1)+1)\cdot\log (n\cdot(\alpha-1)+1)}{\alpha-1}-\frac{\alpha\log(\alpha)}{(\alpha-1)^2}  \\
  && + \left. \frac{(n\cdot(\alpha-1)+1)\cdot\alpha\log(\alpha)}{(\alpha-1)^2}\right) \\
  &=& \Omega(n\log(n)-n\cdot\log(n\cdot(\alpha-1)+1) \\
  && -\frac{\log(n\cdot(\alpha-1)+1)}{\alpha-1} \\
  && - \left. \frac{\alpha\log(\alpha)}{(\alpha-1)^2}+\frac{n\cdot\alpha\log(\alpha)}{\alpha-1}+\frac{\alpha\log(\alpha)}{(\alpha-1)^2}\right)\\
  &=& \Omega(n\log(n)-n\cdot\log(n\cdot(\alpha-1)+1) \\
  && \left. -\frac{\log(n\cdot(\alpha-1)+1)}{\alpha-1} + \frac{n\cdot\alpha\log(\alpha)}{\alpha-1}\right)\\
  &=& \Omega\left(n\log\left(\frac{n}{n\cdot(\alpha-1)+1}\right)-\frac{\log(n\cdot(\alpha-1)+1)}{\alpha-1} \right. \\
  && \left.+\frac{n\cdot\alpha\log(\alpha)}{\alpha-1}\right)\\
  &\subseteq& \Omega\left(n\log\left(\frac{n}{n\cdot(\alpha-1)+1}\right)+\frac{n\cdot\alpha\log(\alpha)}{\alpha-1}\right) \\
  && \text{by \thref{lemma:num-layers-in-little-oh-n}} \\
  &=& \Omega\left(n\log\left(\frac{1}{\alpha-1}\right)+\frac{n\cdot\alpha\log(\alpha)}{\alpha-1}\right) \text{by \thref{lemma:asymptotic-bound}}
\end{eqnarray*}
\end{proof}

In some applications, it may be useful to have a bound on LOHification
that includes $\alpha = 1$.

\begin{theorem}\thlabel{thm:full-lower-bound-on-lohify}
  $\forall$ $\alpha\geq 1$, LOHification $\in\Omega(n\log(\frac{n}{n\cdot(\alpha-1)+1})+\frac{n\cdot\alpha\cdot\log(\alpha)}{\alpha-1})$
\end{theorem}
\begin{proof}
  From our proof in \thref{thm:lower-bound-on-lohify}, we know that
  LOHification with $\alpha>1$ is in $\Omega(n\log(\frac{n}
  {n\cdot(\alpha-1)+1})+\frac{n\cdot\alpha\cdot\log(\alpha)}{\alpha-1})$.
  Because LOHification with $\alpha=1$ is sorting, which is in $\Omega(n\log(n))$,
  it just remains to show that our bound is $\Omega(n\log(n))$ at $\alpha=1$.
\begin{eqnarray*}
  &&  \Omega\left(\lim\limits_{\alpha\rightarrow 1} n\log\left(\frac{n}{n\cdot(\alpha-1)+1}\right)+\frac{n\cdot\alpha\log(\alpha)}{\alpha-1}\right) \\
  &=& \Omega\left(n\log(n) + \lim\limits_{\alpha\rightarrow 1} \frac{n\cdot\alpha\log(\alpha)}{\alpha-1}\right)\\
  &=& \Omega\left(n\log(n) + \lim\limits_{\alpha\rightarrow 1} \frac{n\cdot(\log(\alpha)+1)}{1}\right) \\
  && \text{~by L'H{\^o}pital's rule}\\
  &=& \Omega\left(n\log(n) + n \right)\\
  &=& \Omega(n\log(n))
\end{eqnarray*}
\end{proof}

We have now established LOHification to be in $\Omega(n\log(\frac{n}
{n\cdot(\alpha-1)+1})+\frac{n\cdot\alpha\cdot\log(\alpha)}{\alpha-1})$ for
any $\alpha$ at least 1. In the following sections, we will explore different
algorithms for LOHification, their complexity, and for what values of
$\alpha$ they are optimal.

\subsection{LOHification via sorting}
Sorting in $\Theta(n \log(n))$ trivially LOHifies an array (sorting can
be done using any LOHification method by setting $\alpha=1$); note
that this also guarantees the LOH property for any $\alpha\geq 1$, because
any partitioning of layers in a sorted array will have
$L_i \leq L_{i+1}$. Hence, LOHification is in $O(n\log(n))$.

\subsubsection{When sorting is optimal}
$\alpha=1$ indicates each layer has $|L_i|=1$, meaning an ordering
over all elements; this means that sorting must be performed. Thus,
for $\alpha=1$, sorting is optimal. Furthermore, we can find
an $\alpha^*$ where sorting is optimal for all $\alpha\leq\alpha^*$.
Doing this, we find that, for any constant, $C > 0$, sorting is
optimal for $\alpha^* \leq 1+\frac{C}{n}$.

\begin{theorem}\thlabel{thm:optimal-alpha-for-sorting}
  For any constant, $C > 0$, sorting is optimal for $\alpha\leq\left(1+\frac{C}{n}\right):=\alpha^*$
\end{theorem}
\begin{proof}
  Because decreasing $\alpha$ can only increase the number of layers (and therefore the work), its suffices to show that sorting is optimal at $\alpha^*=\left(1+\frac{C}{n}\right)$.

\begin{eqnarray*}
  \text{r($n$)} &\in& \Omega\left(n\log\left(\frac{1}{\alpha^*-1}\right)+\frac{n\cdot\alpha^*\cdot\log(\alpha^*)}{\alpha^*-1}\right) \\
  &=& \Omega\left(n\log\left(\frac{1}{\left(1+\frac{C}{n}\right)-1}\right)\right.\\
  && \left. +\frac{n\cdot\left(1+\frac{C}{n}\right)\cdot\log(\left(1+\frac{C}{n}\right))}{\left(1+\frac{C}{n}\right)-1}\right) \\
  &=& \Omega\left(n\log\left(\frac{n}{C}\right)+\frac{(n+C)\cdot\log\left(1+\frac{C}{n}\right)}{\frac{C}{n}}\right) \\
  &=& \Omega(n\log(n)-n\log(C) \\
  && \left. +\frac{(n^2+C\cdot n)\cdot\log\left(1+\frac{C}{n}\right)}{C}\right) \\
  &=& \Omega\left(n\cdot\log(n) + (n^2+n)\cdot\log\left(1+\frac{C}{n}\right)\right) \\
  &=& \Omega(n\cdot\log(n) + o(n\cdot\log(n))) \text{~\thref{lemma:optimal-alpha-for-sorting}} \\
  &\subseteq& \Omega(n\cdot\log(n)) \\
  && \text{Therefore;} \\
  \text{LOH} &\in& \Theta(n\cdot\log(n))~~\forall~~\alpha\leq\left(1+\frac{C}{n}\right)
\end{eqnarray*}
\end{proof}

Because sorting is optimal for these values of $\alpha$, we know
that, for all $\alpha$ at most $a^*=\left(1+\frac{C}{n}\right)$, LOHification
is in $\Theta(n\log(\frac{n}{n\cdot(\alpha-1)+1})+\frac{n\cdot\alpha\cdot
\log(\alpha)}{\alpha-1})$. Next we will look at LOHification methods that
are based on selection.

\subsection{LOHification via iterative selection}
LOHs can be constructed using one-dimensional selection
(one-dimensional selection can be done in linear time via
median-of-medians\cite{blum:time}). In this section, we
will describe a LOHification algorithm that selects away
layers from the end of the array, prove its complexity,
and find for which values of $\alpha$ it is optimal.

\subsubsection{Selecting away the layer with the greatest index}

This algorithm repeatedly performs a linear-time one-dimensional
selection on the value at the
first index (were the array in sorted order) in $L_{\ell-1}$, then the
LOH is partitioned about this value. This is repeated for
$L_{\ell-2}$, $L_{\ell-3}$, and so on until the LOH has been
partitioned about the minimum value in each layer. We will prove
that this algorithm is in $\Theta\left(\frac{\alpha\cdot n}{\alpha-1}
\right)$.

\begin{lemma}\thlabel{lemma:lower-bound-on-select-away-greatest-index}
  Selecting away the layer with the greatest index is in $\Omega\left(\frac{\alpha\cdot n}{\alpha-1}\right)$
\end{lemma}
\begin{proof}
  By using a linear time one-dimensional selection,  we can see that
  the runtime for selecting away the layer with the greatest index is:
\begin{eqnarray*}
  && r(n) \\
  &\in& \Theta\left(\sum_{i=0}^{\ell-1}\left( n - \sum_{j=\ell-i}^{\ell-1} |L_j|\right)\right)\\
  &\subseteq& \Omega\left(\sum_{i=0}^{\ell-1}\left( n - \sum_{j=\ell-i}^{\ell-1} \lceil\alpha^j\rceil\right)\right)\\
  &\subseteq& \Omega\left(\sum_{i=0}^{\ell-1}\left( n - \sum_{j=\ell-i}^{\ell-1} (\alpha^{j}+1)\right)\right)\\
  &=& \Omega\left(\sum_{i=0}^{\ell-1}\left( n - i - \sum_{j=\ell-i}^{\ell-1} \alpha^j\right)\right)\\
  &=& \Omega\left(n\cdot \ell - \frac{\ell^{2}-\ell}{2} -\frac{1}{\alpha-1} \cdot \left(\sum_{i=0}^{\ell-1}(\alpha^{\ell}-\alpha^{\ell-i})\right)\right)\\
  &=& \Omega\left(n\cdot \ell - \frac{\ell^{2}-\ell}{2} -\frac{1}{\alpha-1} \right. \\
  && \left. \cdot \frac{(\ell-1)\alpha^{\ell+1}-\ell\cdot\alpha^\ell+\alpha}{\alpha-1}\right)\\
  &=& \Omega\left(n\cdot \ell - \frac{\ell^{2}-\ell}{2} -\frac{1}{\alpha-1} \right.\\
  && \left. \cdot \frac{(\alpha-1)\cdot\ell\cdot\alpha^{\ell}-\alpha\cdot(\alpha^{\ell}-1)}{\alpha-1}\right)\\
  &=& \Omega\left(n\cdot \ell - \frac{\ell^{2}-\ell}{2} -\frac{\ell\cdot\alpha^\ell}{\alpha-1} + \frac{\alpha\cdot(\alpha^{\ell}-1)}{(\alpha-1)^2} \right)\\
  &=& \Omega\left(\ell\cdot\left(n - \frac{\ell-1}{2} -\frac{\alpha^\ell}{\alpha-1}\right) + \frac{\alpha\cdot(\alpha^{\ell}-1)}{(\alpha-1)^2} \right)\\
  &\subseteq& \Omega\left(\ell\cdot\left(n - \frac{\ell-1}{2} -\frac{n\cdot(\alpha-1)+1}{\alpha-1}\right) \right. \\
  && \left. + \frac{\alpha\cdot n\cdot(\alpha-1)}{(\alpha-1)^2} \right) \text{by ~\thref{lemma:lower-bound-on-number-layers}}\\
  &=& \Omega\left(\ell\cdot\left( - \frac{\ell-1}{2} -\frac{1}{\alpha-1}\right) + \frac{\alpha\cdot n}{\alpha-1}\right) \\
  &=& \Omega\left(\frac{\alpha\cdot n}{\alpha-1} - \left(\frac{\ell^{2}-\ell}{2}+\frac{\ell}{\alpha-1}\right)\right)\\
  &\subseteq& \Omega\left(\frac{\alpha\cdot n}{\alpha-1} - \right. \\
  && \left(\frac{(\log_{\alpha}(n\cdot(\alpha-1)+1))^{2} - \log_{\alpha}(n\cdot(\alpha-1))}{2}\right) \\
  && - \left.\left(\frac{\log_{\alpha}(n\cdot(\alpha-1))}{\alpha-1} \right)\right) \\
  && \text{by ~\thref{lemma:upper-bound-on-number-layers} and \thref{lemma:lower-bound-on-number-layers}}\\
  &\subseteq& \Omega\left(\frac{\alpha\cdot n}{\alpha-1}\right) \text{~\thref{lemma:previously-lemma-5}  and \thref{lemma:previously-lemma-6}}
\end{eqnarray*}
\end{proof}

\begin{theorem}\thlabel{thm:tight-bounds-on-select-away-greatest-index}
  Selecting away the layer with the greatest index is in $\Theta\left(\frac{\alpha\cdot n}{\alpha-1}\right)$
\end{theorem}
\begin{proof}
  Using a linear time one-dimensional selection, we can see that the
  runtime for selecting away the layer with the greatest index is:
\begin{eqnarray*}
  && r(n) \\ 
  &\in& \Theta\left(\sum_{i=0}^{\ell-1}\left( n - \sum_{j=\ell-i}^{\ell-1} |L_j|\right)\right)\\
  &\subseteq& O\left(\sum_{i=0}^{\ell-1}\left( n - \sum_{j=\ell-i}^{\ell-1} \alpha^j \right)\right)\\
  &=& O\left(n\cdot \ell -\frac{1}{\alpha-1} \cdot \left(\sum_{i=0}^{\ell-1}(\alpha^{\ell}-\alpha^{\ell-i})\right)\right)\\
  &=& O\left(n\cdot \ell -\frac{1}{\alpha-1} \cdot \frac{(\ell-1)\alpha^{\ell+1}-\ell\cdot\alpha^\ell+\alpha}{\alpha-1}\right)\\
  &=& O\left(n\cdot \ell -\frac{1}{\alpha-1} \cdot \frac{(\alpha-1)\cdot\ell\cdot\alpha^{\ell}-\alpha\cdot(\alpha^{\ell}-1)}{\alpha-1}\right)\\
  &=& O\left(n\cdot \ell -\frac{\ell\cdot\alpha^\ell}{\alpha-1} + \frac{\alpha\cdot(\alpha^{\ell}-1)}{(\alpha-1)^2}\right)\\
  &=& O\left(\ell\cdot\left(n  -\frac{\alpha^\ell}{\alpha-1}\right) + \frac{\alpha\cdot(\alpha^{\ell}-1)}{(\alpha-1)^2}\right)\\
  &\subseteq& O\left(\ell\cdot\left(n -\frac{n\cdot(\alpha-1)+1}{\alpha-1}\right) + \frac{\alpha\cdot n\cdot(\alpha-1)}{(\alpha-1)^2}\right) \\
  && \text{by ~\thref{lemma:lower-bound-on-number-layers}}\\
  &=& O\left(\ell\cdot\left(-\frac{1}{\alpha-1}\right) + \frac{\alpha\cdot n}{\alpha-1}\right)\\
  &=& O\left(\frac{\alpha\cdot n}{\alpha-1} - \left(\frac{\ell}{\alpha-1}\right)\right) \\
  &\subseteq& O\left(\frac{\alpha\cdot n}{\alpha-1} - \left(\frac{\log_{\alpha}(n\cdot(\alpha-1))}{\alpha-1} \right)\right) \\
  && \text{by ~\thref{lemma:lower-bound-on-number-layers}}\\
  &\subseteq& O\left(\frac{\alpha\cdot n}{\alpha-1}\right) \text{by ~\thref{lemma:previously-lemma-6}}\\
  && \text{therefore by \thref{lemma:lower-bound-on-select-away-greatest-index};} \\
  && r(n) \in \Theta\left(\frac{\alpha\cdot n}{\alpha-1}\right)
\end{eqnarray*}
\end{proof}

\subsubsection{When iterative selection is optimal}
We shall also assume that $\alpha>1$ as sorting is optimal for
$\alpha=1$. We will prove that this method is optimal for all values of
$\alpha$ at least two, but not for all values of $\alpha$ less than two.

\begin{theorem}\thlabel{thm:when-iterative-selection-is-optimal}
  Iterative selection is optimal for all $\alpha\geq 2$
\end{theorem}
\begin{proof}
  LOHification is trivially done in $\Omega(n)$, as that is the cost
  to load the data. As $\alpha$ increases, the number of layers (hence
  the work) can only decrease, thus it suffices to show iterative
  selection is optimal at $\alpha=2$.
\begin{eqnarray*}
  r(n) &\in& O\left(\frac{2\cdot n}{2-1}\right) \text{~\thref{thm:tight-bounds-on-select-away-greatest-index}}\\
  &\in& O(n) \\
  && \text{therefore;} \\
  \text{LOH} &\in& \Theta(n) ~~ \forall~~\alpha\geq 2
\end{eqnarray*}
\end{proof}

\begin{lemma}\thlabel{lemma:alpha-for-which-iterative-selection-is-sub-optimal}
  Iterative selection is sub-optimal for
  $\alpha=\alpha^*=1+\frac{C}{n}$ where $C$ is any constant $>0$.
\end{lemma}
\begin{proof}
\begin{eqnarray*}
  r(n) &\in& \Theta\left(\frac{\alpha^*\cdot n}{\alpha^*-1}\right) \text{~\thref{thm:tight-bounds-on-select-away-greatest-index}}\\
  &=& \Theta\left(\frac{\left(1+\frac{C}{n}\right)\cdot n}{\left(1+\frac{C}{n}\right)-1}\right) \\
  &=& \Theta\left(\frac{n+C}{\left(\frac{C}{n}\right)}\right) \\
  &=& \Theta\left(\frac{n^2+C\cdot n}{C}\right) \\
  &\subseteq& \Theta(n^2) \\
  &\subseteq& \omega(n\cdot\log(n))
\end{eqnarray*}
\end{proof}

\begin{theorem}\thlabel{alpha-interval-for-which-iterative-selection-is-sub-optimal}
  Iterative selection is sub-optimal for $1<\alpha<2$
\end{theorem}
\begin{proof}
For this derivation, we shall look at the runtime of iterative
selection as a function of $\alpha$ defined by
$f(\alpha) = \frac{\alpha\cdot n}{\alpha-1}$. We can see that
$f'(\alpha) = \frac{-n}{(\alpha-1)^2}$ is negative for all
$\alpha > 1$, thus it is decreasing on the interval
$\alpha$ in $(1,\infty)$. Because decreasing $\alpha$ can only increase
the number of layers (hence the runtime), we know the runtime is
sub-optimal for $\alpha\leq\alpha^*$ by
\thref{lemma:alpha-for-which-iterative-selection-is-sub-optimal}.
Because $f(\alpha) = \frac{\alpha\cdot n}{\alpha-1}$ is continuous and
decreasing on the interval $\alpha$ in $(1,\infty)$ and sub-optimal at
$\alpha=\alpha^*$; it is sub-optimal for $\alpha^*\leq\alpha<\alpha'$
where $\alpha'$ is the first value of $\alpha$, greater than 1, for which
$f(\alpha) = \frac{\alpha\cdot n}{\alpha-1}$ is optimal. We can find
$\alpha'$ by solving:
\begin{eqnarray*}
  \frac{\alpha'\cdot n}{\alpha'-1} &=& n\log_2\left(\frac{1}{\alpha'-1}\right)+\frac{n\cdot\alpha'\cdot\log_2(\alpha')}{\alpha'-1}
\end{eqnarray*}
Which can be simplified to:
\begin{eqnarray*}
  \frac{\alpha'}{\alpha'-1} &=& \log_2\left(\frac{1}{\alpha'-1}\right)+\frac{\alpha'\cdot\log_2(\alpha')}{\alpha'-1} 
\end{eqnarray*}
We see that $\alpha'=2$ is our solution. Therefore, iterative selection is sub-optimal for $1<\alpha<2$.
\end{proof}

\subsection{Selecting to divide remaining pivot indices in half}
For this algorithm, we first calculate the pivot indices in $O(n)$. Then,
we perform a linear-time one-dimensional selection
on the layers up to the median pivot. We then recurse on the sub-problems
until the array is LOHified.

\subsubsection{Runtime}
Because one-dimensional selection is in $\Theta(n)$, the cost of every
layer in the recursion is in $\Theta(n)$. Because splitting at the
median pivot creates a balanced-binary recursion tree, the cost of
the algorithm is in $\Theta(n\cdot d)$ where $d$ is the depth of the
recursion tree. Because the number of pivots in each recursive call
is one less than half of the number of pivots in the parent call, we
have $d=\log_2(\ell)$. Hence:
\begin{eqnarray*}
  r(n) &\in& \Theta(n\cdot\log(\ell)) \\
  &=& \Theta(n\cdot\log(\log_{\alpha}(n\cdot(\alpha-1)+1))) \\
  &=& \Theta\left(n\cdot\log\left(\frac{\log(n\cdot(\alpha-1)+1)}{\log(\alpha)}\right)\right)
\end{eqnarray*}

\subsubsection{When selecting to divide remaining pivot indices in half is optimal}
Here we will show that this method is optimal for the values of
$\alpha$ where sorting is optimal, i.e. $1\leq\alpha\leq\alpha^*=1+\frac{C}{n}$
for any constant, $C > 0$. Then, however, we will show that it is not optimal
for some interval between $\alpha^*$ and two.

\begin{lemma}\thlabel{lemma:optimal-alpha-for-which-selecting-to-div-pivot-indices-in-half-is-optimal}
  Selecting to divide remaining pivot indices in half is optimal for $\alpha=\alpha^*=1+\frac{C}{n}$ for any constant, $C > 0$.
\end{lemma}
\begin{proof}
\begin{eqnarray*}
  r(n) &\in& \Theta\left(n\cdot\log\left(\frac{\log(n\cdot(\alpha^*-1)+1)}{\log(\alpha^*)}\right)\right) \\
  &=& \Theta\left(n\cdot\log\left(\frac{\log\left(n\cdot\left(\left(1+\frac{C}{n}\right)-1\right)+1\right)}{\log\left(1+\frac{C}{n}\right)}\right)\right)\\
  &=& \Theta\left(n\cdot\log\left(\frac{\log(C))}{\log\left(1+\frac{C}{n}\right)}\right)\right) \\
  &=& \Theta(n\cdot\log(n)) \text{~\thref{lemma:lemma-previously-labeled-9}}
\end{eqnarray*}
\end{proof}

\begin{lemma}\thlabel{lemma:when-select-to-divide-pivots-is-sub-optimal}
  Selecting to divide remaining pivot indices in half is sub-optimal for $\alpha=2$
\end{lemma}
\begin{proof}
\begin{eqnarray*}
  r(n) &\in& \Theta\left(n\cdot\log\left(\frac{\log(n\cdot(2-1)+1)}{\log(2)}\right)\right) \\
  &=& \Theta(n\cdot\log(\log(n+1))) \\
  &\subseteq& \Theta(n\cdot\log(\log(n))) \\
  &\subseteq& \omega(n) 
\end{eqnarray*}
\end{proof}

\begin{theorem}\thlabel{thm:when-selecting-to-divide-remaining-pivot-indices-in-half-is-sub-optimal}
  Selecting to divide remaining pivot indices in half is sub-optimal for some interval in $\alpha^*<\alpha\leq 2$
\end{theorem}
\begin{proof}
For this derivation, we shall look at the runtime of dividing the remaining pivot
indices in half as a function of $\alpha$ defined by
$f(\alpha)=n\cdot\log\left(\frac{\log(n\cdot(\alpha-1)+1)}{\log(\alpha)}\right)$.
By \thref{lemma:optimal-alpha-for-which-selecting-to-div-pivot-indices-in-half-is-optimal}
and \thref{lemma:when-select-to-divide-pivots-is-sub-optimal}, $f(\alpha)$ is optimal
at $\alpha^*$ and sub-optimal at $2$. Because
\[f'(\alpha)=
\frac{n\cdot\log(\alpha)\cdot\left(
  \frac{n\cdot\alpha\cdot\log(\alpha)-(n\cdot(\alpha-1)+1)\cdot\log(n\cdot(\alpha-1)+1)}
       {(n\cdot(\alpha-1)+1)\cdot\log^2(\alpha)\cdot\alpha}\right)}
{\log(n\cdot(\alpha-1)+1)}\]
is negative for large $n$ and $\alpha>1$, the algorithm performs better as
$\alpha$ increases. Because it is sub-optimal at $\alpha=2$ there must be an
interval in $(\alpha^*,2]$ where $f(\alpha)$ is sub-optimal.
\end{proof}

\subsection{Partitioning on the pivot closest to the center of the array}
For this implementation of the algorithm, we start by computing the pivots and
then performing a linear-time selection algorithm on the pivot closest to the
true median of the array to partition the array into two parts. We then recurse
on the parts until all layers are generated. In this section, we will describe
the runtime recurrence in detail, and then prove that this method has optimal
performance at any $\alpha$.

\subsubsection{The runtime recurrence}
Let $n_s$ and $n_e$ be the starting and ending indices (respectively) of our (sub)array. 
Let $m(n_s,n_e)$ be the number of pivots between $n_s$ and $n_e$ (exclusive). 
Let $x(n_s,n_e)$ be the index of the pivot closest to the middle of the (sub)array starting at $n_s$ and ending at $n_e$. 
Then the runtime of our algorithm is $r(0,n)$ where \\
\[
  r(n_s,n_e) = \begin{cases}
    0,~~~~ n_s \geq n_e\\
    0,~~~~  m(n_s,n_e) = 0 \\
    n_e-n_s+r(n_s,x(n_s,n_e)-1)+ \\ r(x(n_s,n_e)+1,n_e),~~~~ \text{else}
  \end{cases}
\]

The recurrence for this algorithm is solved by neither the master
theorem\cite{bently:general} nor the more general Akra-Bazzi
method\cite{akra:solution}. Instead, we will bound the runtime by
bounding how far right we go in the recursion tree, $t_{max}$, and
using this to find the deepest layer, $d^*$ for which all branches
have work.  Because performing two selections is in $O(n)$, we will
bound the size of the recursions by half of the parent by selecting on
the pivots on both sides of the true median (if the true median is a
pivot we just pay for it twice). From there, the bound on the runtime
can be computed as $O(d^*\cdot n) +
O\left(\sum^{\log(n)}_{d=d^{*}}\sum^{t_{max}}_{t=1}\frac{n}{2^{d}}\right)$.
This scheme is depicted in Figure~\ref{fig:tree}.

\begin{figure}[H]
  \centering
  \includegraphics[width=.45\textwidth]{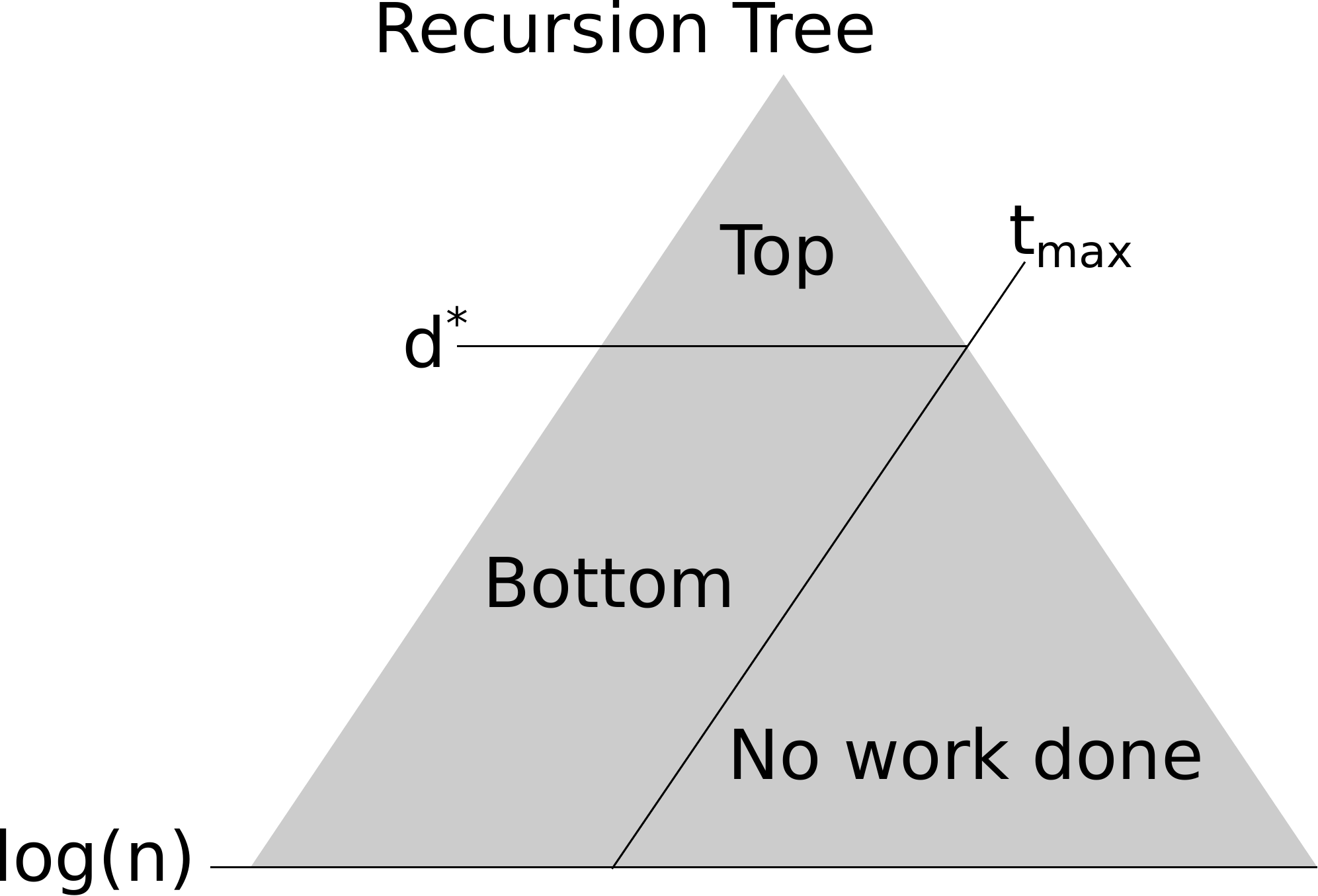}
  \caption{\textbf{The recursion tree for partitioning on the pivot
    closest to the center of the array.}
    The work at ``Top'' is in $O(n\cdot d^*)$ and the work done at
    ``Bottom'' is in $O\left(\sum^{\log(n)}_{d=d^{*}}\sum^{t_{max}}_{t=1}
    \frac{n}{2^{d}}\right)$.
  \label{fig:tree}}
\end{figure}

\subsubsection{Bounds on variables}

For the derivation of the following bounds, we shall assume that
$\alpha > 1$.

\begin{lemma}\thlabel{lemma:bound-on-number-pivots-between-two-points}
  The number of pivots between any two points is 
  $m(n_s,n_e) \leq \log_{\alpha}\left(\frac{n_{e}\cdot(\alpha-1)+1}{(n_{s}-1)\cdot(\alpha-1)+1}\right)$.
\end{lemma}
\begin{proof}
By our definition, the $i^{th}$ pivot, $p_i$, occurs at $p_{i} = \left\lceil\sum_{j=0}^{i}\alpha^j\right\rceil = \left\lceil\frac{\alpha^{i+1}-1}{\alpha-1}\right\rceil$.
Let $n_s$ be the start of our (sub)array and $n_e$ be the end of our (sub)array. Then the number of pivots, $p_e$, occurring before $n_e$ is bound by the inequality:\\
\begin{eqnarray*}
  n_e &\geq& \left\lceil\frac{\alpha^{p_{e}+1}-1}{\alpha-1}\right\rceil \\
  n_e &\geq& \frac{\alpha^{p_{e}+1}-1}{\alpha-1} \\
  n_{e}\cdot(\alpha-1) &\geq& \alpha^{p_{e}+1}-1 \\
  n_{e}\cdot(\alpha-1)+1 &\geq& \alpha^{p_{e}+1} \\
  \log_{\alpha}(n_{e}\cdot(\alpha-1)+1) &\geq& p_{e}+1 \\
  \log_{\alpha}(n_{e}\cdot(\alpha-1)+1)-1 &\geq& p_{e}
\end{eqnarray*}
Similarly, the number of pivots, $p_s$, occurring before $n_s$ is bound by the inequality:
\begin{eqnarray*}
  n_s &\leq& \left\lceil\frac{\alpha^{p_{s}+1}-1}{\alpha-1}\right\rceil \\
  n_s &\leq& \frac{\alpha^{p_{s}+1}-1}{\alpha-1} + 1 \\
  n_{s}-1 &\leq& \frac{\alpha^{p_{s}+1}-1}{\alpha-1} \\
  (n_{s}-1)\cdot(\alpha-1) &\leq& \alpha^{p_{s}+1}-1 \\
  (n_{s}-1)\cdot(\alpha-1)+1 &\leq& \alpha^{p_{s}+1} \\
  \log_{\alpha}((n_{s}-1)\cdot(\alpha-1)+1) &\leq& p_{s}+1 \\
  \log_{\alpha}((n_{s}-1)\cdot(\alpha-1)+1)-1 &\leq& p_{s}
\end{eqnarray*}
By combining these two inequalities, we can find an upper bound on the number of pivots in the (sub)array, $m(n_s,n_e)$:
\begin{eqnarray*}
  m(n_s,n_e) &\leq& (\log_{\alpha}(n_{e}\cdot(\alpha-1)+1)-1) \\
  && - (\log_{\alpha}((n_{s}-1)\cdot(\alpha-1)+1)-1) \\
  m(n_s,n_e) &\leq& \log_{\alpha}(n_{e}\cdot(\alpha-1)+1) \\
  && - \log_{\alpha}((n_{s}-1)\cdot(\alpha-1)+1) \\
  m(n_s,n_e) &\leq& \log_{\alpha}\left(\frac{n_{e}\cdot(\alpha-1)+1}{(n_{s}-1)\cdot(\alpha-1)+1}\right) 
\end{eqnarray*}
\end{proof}

\subsubsection{A bound on the runtime recurrence}
For the following bounds, we will assume that $\alpha > 1$.
Let $d$ be the depth of our current recursion (indexed at 0) and $t$
be how far right in the tree we are at our current recursion (indexed
at 1). To get an upper bound on the recurrence, we will compute the
cost of selecting for both the first index before the true middle and
the first index after the true middle. We will then treat the true
middle as $x(n_s,n_e)$ for our recursive calls. Under these
restrictions, $n_s=\frac{n\cdot(t-1)}{2^d}$ and
$n_e=\frac{n\cdot t}{2^d}$ for a given $t$ and $d$. Knowing this, we
can calculate bounds for $m(n_s,n_e)$ in terms of $t$ and $d$.
\begin{eqnarray*}
  m(n_s,n_e) &\leq& \log_{\alpha}\left(\frac{n_{e}\cdot(\alpha-1)+1}{(n_{s}-1)\cdot(\alpha-1)+1}\right) \\
  && \text{by ~\thref{lemma:bound-on-number-pivots-between-two-points}}\\
  &\leq& \log_{\alpha}\left(\frac{\frac{n\cdot t}{2^d}\cdot(\alpha-1)+1}{(\frac{n\cdot(t-1)}{2^d}-1)\cdot(\alpha-1)+1}\right) \\
  &\leq& \log_{\alpha}\left(\frac{n\cdot t\cdot(\alpha-1)+2^d}{(n\cdot(t-1)-2^d)\cdot(\alpha-1)+2^d}\right)
\end{eqnarray*}

We can then use this to calculate $t$, in terms of $\alpha$, $n$ and $d$
for which $m(n_s,n_e) < 1$. This will give us a bound on how far right
we go in the recursion tree.

\begin{eqnarray*}
  && \log_{\alpha}\left(\frac{n\cdot t\cdot(\alpha-1)+2^d}{(n\cdot(t-1)-2^d)\cdot(\alpha-1)+2^d}\right) < 1 \\
  && \left(\frac{n\cdot t\cdot(\alpha-1)+2^d}{(n\cdot(t-1)-2^d)\cdot(\alpha-1)+2^d}\right) < \alpha \\
  && n\cdot t\cdot(\alpha-1)+2^d < \alpha\cdot(n\cdot(t-1)-2^d)\cdot(\alpha-1) \\
  && ~~~~~~~~~~~~~~~~~~~~~~~~~~+\alpha\cdot 2^d \\
  && n\cdot t\cdot(\alpha-1) < \alpha\cdot(n\cdot(t-1)-2^d)\cdot(\alpha-1) \\
  && ~~~~~~~~~~~~~~~~~~~~+(\alpha-1)\cdot 2^d \\
  && n\cdot t < \alpha\cdot(n\cdot(t-1)-2^d)+ 2^d \\
  && n\cdot t < \alpha\cdot n\cdot t-\alpha\cdot n-\alpha\cdot 2^d+2^d \\
  && t < \alpha\cdot t-\alpha-\frac{(\alpha-1)\cdot 2^d}{n} \\
  && t-\alpha\cdot t < -\alpha-\frac{(\alpha-1)\cdot 2^d}{n} \\
  && t\cdot(\alpha-1) > \alpha+\frac{(\alpha-1)\cdot 2^d}{n} \\
  && t > \frac{\alpha}{\alpha-1}+\frac{2^d}{n}
\end{eqnarray*}

\indent Because $2^d\leq n$ at any layer of the recursion,
$t_{max}=\frac{\alpha}{\alpha-1}+1$. Using this, we can define
$r^*$, an upper bound on our runtime recurrence where
$r(0,n) \leq r^*(1,0)$ and $r^*(t,d) =$
\[
  \begin{cases}
    0 & t > t_{\max}\\
    0 & 2^d > n\\
    \frac{n}{2^d}+r(2\cdot t-1,d+1)+r(2\cdot t,d+1) & \text{else}.
  \end{cases}
\]

\subsubsection{The runtime of partitioning on the pivot closest to the center of the array}

\begin{theorem}\thlabel{thm:time-to-divide-array-in-half}
  For $\alpha>1$, partitioning on the pivot closest to the center of the array is
  in $O\left(n\log\left(\frac{\alpha}{\alpha-1}\right)\right)$
\end{theorem}
\begin{proof}
Let $d^*$ be the largest d for which all branches at layer d have work.
Because $t_{max}=\frac{\alpha}{\alpha-1}+1$, $d^*=\log_2(\frac{\alpha}{\alpha-1}+1)$.
This yields:
\begin{eqnarray*}
  && r(n) \\
  &\leq& r^*(n) \\
  &\in& O\left(\sum^{\log(n)}_{d=d^{*}}\sum^{t_{max}}_{t=1}\frac{n}{2^{d}}\right)+O(n\cdot d^*) \\
  &\in& O\left(\sum^{\log(n)}_{d=d^{*}}\frac{\alpha}{\alpha-1}\cdot\frac{n}{2^{d}}\right)+O\left(n\cdot\log\left(\frac{\alpha}{\alpha-1}+1\right)\right) \\
  &\in& O\left(\frac{n\cdot\alpha}{\alpha-1}\cdot\sum^{\log(n)}_{d=d^{*}}\frac{1}{2^{d}}\right)+O\left(n\cdot\log\left(\frac{\alpha}{\alpha-1}\right)\right) \\
  &\in& O\left(\frac{n\cdot\alpha}{\alpha-1}\cdot\left(2^{1-d^*}-2^{-\log(n)}\right)\right) \\
  && +O\left(n\cdot\log\left(\frac{\alpha}{\alpha-1}\right)\right) \\
  &\in& O\left(\frac{n\cdot\alpha}{\alpha-1}\cdot\left(2\cdot\frac{\alpha-1}{2\cdot\alpha-1}-\frac{1}{n}\right)\right) \\
  && +O\left(n\cdot\log\left(\frac{\alpha}{\alpha-1}\right)\right) \\
  &\in& O\left(\frac{2\cdot n\cdot\alpha}{2\cdot\alpha-1}-\frac{\alpha}{\alpha-1}\right)+O\left(n\cdot\log\left(\frac{\alpha}{\alpha-1}\right)\right) \\
  &\in& O(n) + O\left(n\cdot\log\left(\frac{\alpha}{\alpha-1}\right)\right) \\
  &\in& O\left(n\cdot\log\left(\frac{\alpha}{\alpha-1}\right)\right)
\end{eqnarray*}
\end{proof}

\begin{theorem}\thlabel{thm:time-to-divide-array-in-half-sort}
  For $\alpha=1$, partitioning on the pivot closest to the center of the array is optimal.
\end{theorem}
\begin{proof}
Because we are sorting in this case, it suffices to show that this method
is in $O(n\log(n))$. 
Let $d^*$ be the largest d for which all branches at that layer have work.
Because $\alpha=1$, all branches have work.
Thus $d^*=\log_2(n)$.
This yields:
\begin{eqnarray*}
  r(n) &\leq& r^*(n) \\
  &\in& O(n\cdot d^*) \\
  &\in& O(n\log(n))
\end{eqnarray*}
\end{proof}

\begin{lemma}\thlabel{lemma:optimal-alpha-to-divide-array-in-half}
  partitioning on the pivot closest to the center of the array is optimal for
  $1\leq\alpha\leq\alpha^*=1+\frac{C}{n}$ for any constant, $C > 0$.
\end{lemma}
\begin{proof}
  By \thref{thm:time-to-divide-array-in-half-sort}, this method sorts an array
  in $O(n\log(n))$. Because a sorted array is also a LOH of any order and
  LOHification of order $\alpha^*=1+\frac{C}{n}$ for any constant, $C > 0$, is
  in $\Omega(n\log(n))$ by \thref{thm:optimal-alpha-for-sorting}; this method
  is optimal for $1\leq\alpha\leq\alpha^*$.
\end{proof}

\begin{lemma}\thlabel{lemma:optimal-alpha-to-divide-array-in-half-alpha-geq-2}
  partitioning on the pivot closest to the center of the array is optimal for $\alpha\geq 2$
\end{lemma}
\begin{proof}
\begin{eqnarray*}
  r(n) &\in& O\left(n\cdot\log\left(\frac{\alpha}{\alpha-1}\right)\right) \\
  &=& O(n) \\
  &\subseteq& \Theta(n)
\end{eqnarray*}
\end{proof}

\begin{theorem}\thlabel{thm:optimal-at-any-alpha}
  partitioning on the pivot closest to the center of the array is
  optimal for all $\alpha\geq 1$
\end{theorem}
\begin{proof}
By \thref{lemma:optimal-alpha-to-divide-array-in-half} and
\thref{lemma:optimal-alpha-to-divide-array-in-half-alpha-geq-2}, it
suffices to show that partitioning on the pivot closest to the center
of the array is optimal for $\alpha^*<\alpha<2$. Suppose
$\alpha^*<\alpha<2$. Then:
\begin{eqnarray*}
  r(n) &\in& \Omega\left(n\log\left(\frac{1}{\alpha-1}\right)+\frac{n\cdot\alpha\cdot\log(\alpha)}{\alpha-1}\right) \\
  && \text{by ~\thref{thm:lower-bound-on-lohify}} \\
&\in& \Omega\left(n\log\left(\frac{1}{\alpha-1}\right)+n\right) \\
&\subseteq& \Omega\left(n\log\left(\frac{1}{\alpha-1}\right)\right)\\
&& \text{By \text{\thref{thm:time-to-divide-array-in-half}}, we have:}\\
r(n) &\in& O\left(n\cdot\log\left(\frac{\alpha}{\alpha-1}\right)\right)\\
&\in& O\left(n\cdot\log\left(\frac{1}{\alpha-1}\right)+n\cdot\log(\alpha)\right)\\
&\in& O\left(n\cdot\log\left(\frac{1}{\alpha-1}\right)\right) \\
&& \text{hence;} \\
r(n) &\in& \Theta\left(n\cdot\log\left(\frac{1}{\alpha-1}\right)\right) 
\end{eqnarray*}
\end{proof}

\subsection{The optimal runtime for the construction of a layer-ordered heap of any rank}

Partitioning on the pivot closest to the center of the array is optimal for
all $\alpha\geq 1$ by \thref{thm:optimal-at-any-alpha}. We can combine this
with \thref{thm:full-lower-bound-on-lohify} to determine that LOHification
is in:
\begin{eqnarray*}
  \Theta\left(n\log\left(\frac{n}{n\cdot(\alpha-1)+1}\right)
  +\frac{n\cdot\alpha\cdot\log(\alpha)}{\alpha-1}\right)
\end{eqnarray*}

\subsection{Quick LOHify}
For this implementation of the algorithm, we partition on a random element,
record the index of this element in an auxiliary array and then recurse
on the left side until the best element is selected. While this method is
probabilistic with a worst case construction in $O(n^2)$ , it performs
well in practice and has a linear expected construction time.

\subsubsection{Expected Runtime of Quick LOHify}
Quick LOHify can be thought of as a Quick-Selection with $k=1$ and a
constant number of operations per recursion for the auxiliary array.
By this, we know the expected runtime to be in $\Theta(n)$. A direct
proof is also provided.

\begin{theorem}
The expected runtime for Quick LOHify is in $\Theta(n)$
\end{theorem}
\begin{proof}
  The runtime is proportional to the number of comparisons. Suppose $x_i$
  is the $i^{th}$ element in the sorted array and assume without loss of
  generality that $i<j$. We compare $x_i$ and $x_j$ only when one of these
  values is the pivot element. This makes the greatest possible
  probability that these elements are compared $\frac{2}{j}$ as j is the
  minimum range that contains these elements. The expected number of
  comparisons can be found by summing this probability over all pairs of
  elements. This yields:
  \begin{eqnarray*}
    \mathbb{E} &=& \sum_{i=0}^{n-2}\sum_{j=i+1}^{n-1}\frac{2}{j} \\
    &=& 2\cdot\sum_{i=0}^{n-2}\sum_{j=i+1}^{n-1}\frac{1}{j} \\
    &=& 2\cdot \left( \begin{tabular}{lllll}
      1 &+ $\frac{1}{2}$ &+ $\frac{1}{3}$ &+ $\cdots$ &+ $\frac{1}{n-1}$ \\
      &+ $\frac{1}{2}$ &+ $\frac{1}{3}$ &+ $\cdots$ &+ $\frac{1}{n-1}$ \\
      &  &+ $\frac{1}{3}$ &+ $\cdots$ &+ $\frac{1}{n-1}$ \\
      &  &  &+ $\ddots$ &+ $~~\vdots$ \\
      &  &  &  &+ $\frac{1}{n-1}$
      \end{tabular} \right) \\
    &=& 2\cdot(n-1) \\
    &=& 2\cdot n -2\\
    &\in& \Theta(n)
  \end{eqnarray*}
\end{proof}

\subsubsection{Expected $\alpha$ of Quick LOHify}
Unlike other constructions of a LOH, an $\alpha$ is not specified when
performing Quick LOHify nor is it guaranteed to be the same across
different runs. We can, however, determine that the expected value
of $\alpha$ to be in $\Theta(\log(n))$.

\begin{theorem}
The expected $\alpha$ for Quick LOHify is in $\Theta(\log(n))$
\end{theorem}
\begin{proof}
  The average $\alpha$, $\alpha '$, can be computed as the average
  ratio of the last two layers. This can be found by dividing the
  sum of all ratios by the number of ways to choose the pivots.
  This yields:  
  \begin{eqnarray*}
    && \alpha ' \\
    &=& \frac{1}{\binom{n}{2}}\cdot\sum_{i=0}^{n-2}\sum_{j=i+1}^{n-1}\frac{n-j}{j-i}\\
    &=& \frac{2}{n^2-n}\cdot\sum_{i=0}^{n-2}\sum_{j=i+1}^{n-1}\frac{n-j}{j-i}\\
    &=& \frac{2}{n^2-n}\cdot\sum_{i=0}^{n-2}\sum_{k=1}^{n-i-1}\frac{n-i-k}{k}\\
    &=& \frac{2}{n^2-n}\cdot\sum_{i=0}^{n-2}\left(\sum_{k=1}^{n-i-1}\frac{n-i}{k}-1 \right)\\
    &=& \frac{2}{n^2-n}\cdot\sum_{i=0}^{n-2}\left(\left(\sum_{k=1}^{n-i-1}\frac{n-i}{k}\right)-(n-i-2)\right) \\
    &=& \frac{2}{n^2-n}\cdot\sum_{i=0}^{n-2}\left((n-i)\cdot\left(\sum_{k=1}^{n-i-1}\frac{1}{k}\right)-n+i+2 \right) \\
    &=& \frac{2}{n^2-n}\cdot\sum_{i=0}^{n-2}((n-i)\cdot H_{n-i-1}-n+i+2) \\
    &=& \frac{2}{n^2-n} \\
    && \cdot\sum_{i=0}^{n-2}(n\cdot H_{n-i-1}-i\cdot H_{n-i-1}-n+i+2) \\
    &=& \frac{2}{n^2-n} \\
    && \cdot\left(\left(n\cdot\sum_{i=0}^{n-2}(H_{n-i-1})\right)-\left(\sum_{i=0}^{n-2}(i\cdot H_{n-i-1})\right) \right. \\
    && \left. -(n^2-2\cdot n)+\left(\sum_{i=0}^{n-2}i\right)+(2 \cdot n-4) \right) \\
    &=& \frac{2}{n^2-n}\cdot\left(\left(n\cdot\sum_{k=1}^{n-1}H_k\right)-\left(\sum_{k=1}^{n-1}(n-k-1)\cdot H_k\right) \right. \\
    && \left. -(n^2-2\cdot n)+\frac{n^2-3\cdot n+2}{2}+(2\cdot n-4) \right) \\
    &=& \frac{2}{n^2-n}\cdot\left(\left(\sum_{k=1}^{n-1}(k+1)\cdot H_k\right) + \frac{-n^2+5n-6}{2}   \right) \dagger \\
    &=& \frac{2}{n^2-n} \\
    && \cdot\left(\frac{(n^2+n)\cdot H_n}{2}-\frac{n^2}{4}-\frac{3\cdot n}{4}+\frac{-n^2+5\cdot n-6}{2} \right)  \\
    &=& \frac{2\cdot n^2\cdot H_n+2\cdot n\cdot H_n-3\cdot n^2+7\cdot n-12}{2\cdot n^2-2\cdot n} \\
    &\in& \Theta(\log(n)) 
    \end{eqnarray*}
  $\dagger$ Simplified with Wolfram Mathematica \tt{Sum[Sum[(k + 1)/i, {i, 1, k}], {k, 1, -1 + n}]}
\end{proof}

\section{Results}
We compare the runtimes of various LOHify algorithms to compute the most
permissive score threshold at which a given false discovery rate
(FDR)\cite{benjamini:controlling} $\tau$ occurs. This is traditionally
accomplished by sorting the scored hypotheses (which are labeled as TP
or FP) best first and then advancing one element at a time, updating
the FDR to the current $FDR = \frac{\#FP}{\#FP+\#TP}$ at the
threshold, finding the worst score at which $FDR \leq \tau$ occurs.

The LOH method behaves similarly, but they compute optimistic bounds
on the FDRs in each layer (if all TPs in the layer come first) and
pessimistic bounds (if all FPs in the layer come first). When these
bounds include $\tau$, the layer is recursed on, until the size of the
list is in $O(1)$.

Table~\ref{table:fdr-runtimes-size} demonstrates the performance
benefit and the influence of $\alpha$ on practical performance.
\begin{table}[H]
\centering
\footnotesize
\textbf{$\alpha=6.0$ (where applicable)}
\begin{tabular}{r|lllllll}
  $n$ & SORT & SLWGI & SDRPIH & PPCCA & QUICK \\
  \hline
      $2^{28}$  & 27.1712 & 3.60989 & 6.64702 & 3.55166 & 1.20981 \\
      $2^{27}$  & 13.4568 & 2.67432 & 2.98285 & 2.74130 & 0.840145 \\
      $2^{26}$  & 6.44227 & 1.03872 & 1.86184 & 1.05260 & 0.596104 \\
      $2^{25}$  & 3.06724 & 0.58973 & 0.890603& 0.58956 & 0.266691
\end{tabular}
\textbf{Size = $2^{28}$} 
\begin{tabular}{r|lllllll}
  $\alpha$ & SLWGI & SDRPIH & PPCCA  \\
  \hline
      1.05 & 29.9283 & 16.8750 & 14.4878 \\
      1.1  & 19.3675 & 14.8553 & 12.9062 \\
      1.5  & 8.92916 & 11.0468 & 8.12265 \\
      2.0  & 8.24106 & 9.24010 & 8.21261 \\
      3.0  & 5.73344 & 7.94349 & 5.60023 \\
      4.0  & 3.83187 & 6.35753 & 3.84014 \\
      6.0  & 3.60989 & 6.64702 & 3.55166 \\
      8.0  & 4.62627 & 6.90307 & 4.5759 \\
\end{tabular} 
\caption{\textbf{Runtimes (seconds) of different LOHification methods
    with various $\alpha$.}  Reported runtimes are averages over 10
  iterations. SORT is sorting, SLWGI is selecting the layer with the
  greatest index, SDRPIH is selecting to divide the remaining pivot
  indices in half, PPCCA is partitioning on the pivot closest to the
  center of the array, and QUICK is Quick-LOHify. Quick-LOHify
  generates its own partition indices, which are not determined by
  $\alpha$.
  \label{table:fdr-runtimes-size}}
\end{table}

\section{Discussion}
Due to the $\Omega(n \log(n))$ bound on comparison-based sorting, ordering
values using only pairwise comparison is generally considered to be an
area for little practical performance benefit; however, LOHs have been
used to replace sorting in applications where sorting is a limiting factor.
Optimal LOHify for any $\alpha$ and the practically fast Quick-LOHify variant are
useful to replace sorting in applications such as finding the most
abundant isotopes of a compound\cite{kreitzberg:fast} (fast in
practice with $1<\alpha\ll 2$) and finding the score at which a desired FDR
threshold occurs (fast in practice with an $\alpha \gg 2$).

\clearpage

\appendix{\noindent\bf{Appendix}}
\begin{appendices}

\section{Lemmas used in methods}
\begin{lemma}\thlabel{lemma:asymptotic-of-log-of-ceiling}
  $\forall\alpha>1,\lceil\alpha^i\rceil\cdot\log(\lceil\alpha^i\rceil)\sim\alpha^i\cdot\log(\alpha^i)$
\end{lemma}
\begin{proof}
  $\alpha^i\cdot\log(\alpha^i) \leq \lceil\alpha^i\rceil\cdot\log(\lceil\alpha^i\rceil)\leq (\alpha^{i}+1)\cdot\log(\alpha^{i}+1)$
\begin{eqnarray*}
  && \lim\limits_{i\rightarrow\infty} \frac{(\alpha^{i}+1)\cdot\log(\alpha^{i}+1)}{\alpha^i\cdot\log(\alpha^i)} \\
  &=& \lim\limits_{i\rightarrow\infty}  \frac{\alpha^{i}\log(\alpha^{i}+1)}{\alpha^{i}\log(\alpha^i)}+\frac{\log(\alpha^{i}+1)}{\alpha^{i}\log(\alpha^i)}\\
  &=& \lim\limits_{i\rightarrow\infty}  \frac{\log(\alpha^{i}+1)}{\log(\alpha^i)}+\frac{\log(\alpha^{i}+1)}{\alpha^{i}\log(\alpha^i)}\\
  && \text{which, by L'H{\^o}pital's rule;} \\
  &=& \lim\limits_{i\rightarrow\infty}  \frac{\frac{\alpha^{i}\log(\alpha)}{\alpha^{i}+1}}{\log(\alpha)}+\frac{\frac{\alpha^{i}\cdot\log(\alpha)}{\alpha^{i}+1}}{\alpha^{i}\cdot\log(\alpha)\cdot(\log(\alpha^i)+1)}\\
  &=& \lim\limits_{i\rightarrow\infty}  \frac{\alpha^i}{\alpha^{i}+1}+\frac{1}{(\alpha^{i}+1)\cdot(\log(\alpha^i)+1)}\\
  &=& \lim\limits_{i\rightarrow\infty}  1-\frac{1}{\alpha^{i}+1}\\
  &=&  1
\end{eqnarray*}
\end{proof}

\begin{lemma}\thlabel{lemma:num-layers-in-little-oh-n}
  $\forall \alpha > 1, \frac{\log(n\cdot(\alpha-1)+1)}{\alpha-1}\in o(n)$
\end{lemma}
\begin{proof}
\begin{eqnarray*}
  && \lim\limits_{n\rightarrow\infty} \frac{\frac{\log(n\cdot(\alpha-1)+1)}{\alpha-1}}{n} \\
  &=& \lim\limits_{n\rightarrow\infty} \frac{\log(n\cdot(\alpha-1)+1)}{n\cdot(\alpha-1)} \\
  &=& \lim\limits_{n\rightarrow\infty} \frac{\alpha-1}{(\alpha-1)\cdot(n\cdot(\alpha-1)+1)} \text{~by L'H{\^o}pital's rule}\\
  &=& \lim\limits_{n\rightarrow\infty} \frac{1}{n\cdot(\alpha-1)+1} \\
  &=& 0
\end{eqnarray*}
\end{proof}
\begin{lemma}\thlabel{lemma:asymptotic-bound}
  $\forall \alpha > 1, n\log(\frac{n}{n\cdot(\alpha-1)+1}) \sim n\log(\frac{1}{\alpha-1})$
\end{lemma}
\begin{proof}
\begin{eqnarray*}
  && \lim\limits_{n\rightarrow\infty} \frac{n\log\left(\frac{n}{n\cdot(\alpha-1)+1}\right)}{n\log\left(\frac{1}{\alpha-1}\right)} \\
  &=& \lim\limits_{n\rightarrow\infty} \frac{\log\left(\frac{n}{n\cdot(\alpha-1)+1}\right)}{\log\left(\frac{1}{\alpha-1}\right)} \\
  &=& \frac{1}{\log(\frac{1}{\alpha-1})} \cdot \left( \lim\limits_{n\rightarrow\infty} \log\left(\frac{n}{n\cdot(\alpha-1)+1}\right)\right) \\
  &=& \frac{1}{\log(\frac{1}{\alpha-1})} \cdot \left(\log\left( \lim\limits_{n\rightarrow\infty} \frac{n}{n\cdot(\alpha-1)+1}\right)\right) \\
  &=& \frac{1}{\log(\frac{1}{\alpha-1})} \cdot \log(\frac{1}{\alpha-1}) \text{~by L'H{\^o}pital's rule} \\
  &=& 1
\end{eqnarray*}
\end{proof}


\begin{lemma}\thlabel{lemma:optimal-alpha-for-sorting}
  For any constant, $C > 0$, $(n^2+n)\cdot\log(1+\frac{C}{n}) \in o(n\cdot\log(n))$
\end{lemma}
\begin{proof}
\begin{eqnarray*}
  &&  \lim\limits_{n\rightarrow\infty} \frac{(n^2+n)\cdot\log(1+\frac{C}{n})}{n\cdot\log(n)} \\
  &=& \lim\limits_{n\rightarrow\infty} \frac{(n+1)\cdot\log(1+\frac{C}{n})}{\log(n)} \\
  &=& \lim\limits_{n\rightarrow\infty} \frac{n\cdot\log(1+\frac{C}{n})}{\log(n)} + \frac{\log(1+\frac{1}{n})}{\log(n)}\\
  &=& \lim\limits_{n\rightarrow\infty} \frac{n\cdot\log(1+\frac{C}{n})}{\log(n)} \\
  &=& \lim\limits_{n\rightarrow\infty} \frac{\log(1+\frac{C}{n})-\frac{C}{n+C}}{(\frac{C}{n})} \text{~by L'H{\^o}pital's rule}\\
  &=& \lim\limits_{n\rightarrow\infty} \frac{\log(1+\frac{C}{n})}{(\frac{C}{n})}-\frac{n}{n+C} \\
  &=& \lim\limits_{n\rightarrow\infty} \frac{-(\frac{C}{n^2+C\cdot n})}{-(\frac{C}{n^2})} - 1 \text{~by L'H{\^o}pital's rule} \\
  &=& \lim\limits_{n\rightarrow\infty} \frac{n^2}{n^2+C\cdot n} - 1 \\
  &=& \lim\limits_{n\rightarrow\infty} \frac{n}{n+C} - 1 \text{~by L'H{\^o}pital's rule}\\
  &=& 1-1 \text{~by L'H{\^o}pital's rule} \\
  &=& 0
\end{eqnarray*}
\end{proof}


\begin{lemma}\thlabel{lemma:previously-lemma-5}
$\forall\alpha > 1, 
\left(\frac{(\log_{\alpha}(n\cdot(\alpha-1)+1))^{2} - \log_{\alpha}(n\cdot(\alpha-1))}{2} \right)$ $\in o\left(\frac{\alpha\cdot n}{\alpha-1}\right)$
\end{lemma}
\begin{proof}
\begin{eqnarray*}
  && \lim\limits_{n\rightarrow\infty} \frac{\left(\frac{(\log_{\alpha}(n\cdot(\alpha-1)+1))^{2} - \log_{\alpha}(n\cdot(\alpha-1))}{2}\right)}{\left(\frac{\alpha\cdot n}{\alpha-1}\right)} \\
  &=& \lim\limits_{n\rightarrow\infty} \frac{\frac{\alpha-1}{(\log(\alpha))^2}\cdot(\log(n\cdot(\alpha-1)+1))^{2}}{2\cdot\alpha\cdot n} \\
  && - \frac{\frac{\alpha-1}{\log(\alpha)}\cdot\log(n\cdot(\alpha-1)+1)}{2\cdot\alpha\cdot n} \\
  &=& \lim\limits_{n\rightarrow\infty} \frac{\frac{2\cdot(\alpha-1)^2}{(\log(\alpha))^2}\cdot\frac{\log(n\cdot(\alpha-1)+1)}{n\cdot(\alpha-1)+1}-\xcancel{\frac{(\alpha-1)^2}{\log(\alpha)}\cdot\frac{1}{n\cdot(\alpha-1)+1}}}{2\cdot\alpha} \\
  && \text{~by L'H{\^o}pital's rule}\\
  &=& \lim\limits_{n\rightarrow\infty} \frac{(\alpha-1)^2}{\alpha\cdot(\log(\alpha))^2} \cdot \frac{\log(n\cdot(\alpha-1)+1)}{n\cdot(\alpha-1)+1}\\
  &=& \frac{(\alpha-1)^2}{\alpha\cdot(\log(\alpha))^2} \cdot \left(\lim\limits_{n\rightarrow\infty}  \frac{\log(n\cdot(\alpha-1)+1)}{n\cdot(\alpha-1)+1}\right)\\
  &=& \frac{(\alpha-1)^2}{\alpha\cdot(\log(\alpha))^2} \cdot \left(\lim\limits_{n\rightarrow\infty}   \frac{1}{n\cdot(\alpha-1)+1}\right) \\
  && \text{~by L'H{\^o}pital's rule} \\
  &=& 0
\end{eqnarray*}
\end{proof}

\begin{lemma}\thlabel{lemma:previously-lemma-6}
  $\forall\alpha > 1, \left(\frac{\log_{\alpha}(n\cdot(\alpha-1))}{\alpha-1} \right) \in o\left(\frac{\alpha\cdot n}{\alpha-1}\right)$
\end{lemma}
\begin{proof}
\begin{eqnarray*}
  && \lim\limits_{n\rightarrow\infty} \frac{\left(\frac{\log_{\alpha}(n\cdot(\alpha-1))}{\alpha-1} \right)}{\left(\frac{\alpha\cdot n}{\alpha-1}\right)} \\
  &=& \lim\limits_{n\rightarrow\infty} \frac{\log_{\alpha}(n\cdot(\alpha-1))}{\alpha\cdot n} \\
  &=& \lim\limits_{n\rightarrow\infty} \frac{1}{\alpha\cdot\log(\alpha)}\cdot\frac{\log(n\cdot(\alpha-1))}{n} \\
  &=& \frac{1}{\alpha\cdot\log(\alpha)}\cdot\left(\lim\limits_{n\rightarrow\infty} \frac{\log(n\cdot(\alpha-1))}{n}\right) \\
  &=& \frac{1}{\alpha\cdot\log(\alpha)}\cdot\left(\lim\limits_{n\rightarrow\infty} \frac{\alpha}{n\cdot(\alpha-1)}\right) \\
  && \text{~by L'H{\^o}pital's rule}\\
  &=& 0
\end{eqnarray*}
\end{proof}


\begin{lemma}\thlabel{lemma:lemma-previously-labeled-9}
  $n\cdot\log\left(\frac{\log(C)}{\log\left(1+\frac{C}{n}\right)}\right)
  \in\Theta(n\cdot\log(n))$
  \end{lemma}
\begin{proof}
\begin{eqnarray*}
    && \lim\limits_{n\rightarrow\infty} \frac{n\cdot\log\left(\frac{\log(C)}{\log\left(1+\frac{C}{n}\right)}\right)}{n\cdot\log(n)} \\
  &=& \lim\limits_{n\rightarrow\infty} \frac{\log\left(\frac{\log(C)}{\log\left(1+\frac{C}{n}\right)}\right)}{\log(n)} \\
  &=& \lim\limits_{n\rightarrow\infty} \frac{\left(\frac{C}{n\cdot(n+C)\cdot\log\left(1+\frac{C}{n}\right)}\right)}{\left(\frac{1}{n}\right)} \text{~by L'H{\^o}pital's rule}\\
  &=& \lim\limits_{n\rightarrow\infty} \frac{1}{(n+C)\cdot\log\left(1+\frac{C}{n}\right)} \\
  &=& \lim\limits_{n\rightarrow\infty} \frac{\left(\frac{1}{n+C}\right)}{\log\left(1+\frac{C}{n}\right)} \\
  &=& \lim\limits_{n\rightarrow\infty} \frac{\left(\frac{-1}{n^2+2\cdot C\cdot n+C^2}\right)}{\left(\frac{-C}{n^2+C\cdot n}\right)} \text{~by L'H{\^o}pital's rule}\\
  &=& \lim\limits_{n\rightarrow\infty} \frac{n^2+C\cdot n}{C\cdot n^2+2\cdot C^2\cdot n+C^3}\\
  &=& \frac{1}{C} \\
  && \text{~by L'H{\^o}pital's rule}
\end{eqnarray*}
\end{proof}

\end{appendices}

\end{document}